\documentclass[11pt]{article}

\usepackage{amsthm,amsmath,amssymb}
\usepackage[bookmarks=true,hypertexnames=false,pagebackref]{hyperref}
\hypersetup{colorlinks=true, citecolor=blue, linkcolor=red,
  urlcolor=blue}
\usepackage{newpxtext}
\usepackage[libertine,vvarbb]{newtxmath}

\usepackage{fullpage}
\usepackage{hyperref}
\usepackage{cleveref}
\usepackage[utf8]{inputenc}
\usepackage{framed}
\usepackage{enumerate}
\usepackage{graphicx}
\usepackage{float} 
\usepackage{subfigure}

\usepackage[ruled,vlined]{algorithm2e}
\usepackage{algorithmic}

\usepackage{tikz}
\usepackage{complexity}
\usepackage{ulem}

\newtheorem{theorem}{Theorem}[section]

\newtheorem{fact}[theorem]{Fact}
\newtheorem{corollary}[theorem]{Corollary}
\newtheorem{lemma}[theorem]{Lemma}
\newtheorem{definition}[theorem]{Definition}

\newtheorem{remark}[theorem]{Remark}

\newcommand{\ent}{\mathcal{H}}
\newcommand{\mutualent}{\mathcal{I}}

\newcommand{\Yesinstance}{D_{*}}
\newcommand{\Noinstance}{D_{0}}

\newcommand{\ext}{\mathrm{ext}}

\newcommand{\calP}{\mathcal{P}}
\newcommand{\calQ}{\mathcal{Q}}
\newcommand{\calR}{\mathcal{R}}
\newcommand{\calA}{\mathcal{A}}
\newcommand{\calB}{\mathcal{B}}

\newcommand{\calW}{\mathcal{W}}

\newcommand{\Jnote}[1]{[{\color{black}Jiapeng:}{ \color{blue}#1}] }

\newcommand{\kdisj}{k\text{-}\mathrm{UDISJ}}
\newcommand{\sparsedisj}{\mathrm{S}\text{-}\mathrm{UDISJ}}

\title{Lifting Theorems Meet Information Complexity: Known and New Lower Bounds of Set-disjointness\footnote{Working Paper}}

\begin{document}

\author{
Guangxu Yang\thanks{Research supported by NSF CAREER award 2141536.}
\\
Department of Computer Science\\
University of Southern California\\
\texttt{guangxuy@usc.edu}
\and
Jiapeng Zhang\thanks{Research supported by NSF CAREER award 2141536.}
\\
Department of Computer Science\\
University of Southern California\\
\texttt{jiapengz@usc.edu}
}

\maketitle

\begin{abstract}
Set-disjointness problems are one of the most fundamental problems in communication complexity and have been extensively studied in past decades. Given its importance, many lower bound techniques were introduced to prove communication lower bounds of set-disjointness.

Combining ideas from information complexity and query-to-communication lifting theorems, we introduce a density increment argument to prove communication lower bounds for set-disjointness:
\begin{itemize}
\item We give a simple proof showing that a large rectangle cannot be $0$-monochromatic for multi-party unique-disjointness.
\item We interpret the direct-sum argument as a density increment process and give an alternative proof of randomized communication lower bounds for multi-party unique-disjointness.
\item Avoiding full simulations in lifting theorems, we simplify and improve communication lower bounds for sparse unique-disjointness.
\end{itemize}
Potential applications to be unified and improved by our density increment argument are also discussed.
\end{abstract}

\section{Introduction}
\textit{Set-disjointness} is one of the most important problems in communication complexity. Since the formulation of the communication model \cite{yao1979some}, many researchers have made great efforts to understand the communication complexity, both upper and lower bounds, of set-disjointness problems in various communication models \cite{BFS86, KS92,razborov1992on,razborov2003quantum,bar2004information,jain2003lower,haastad2007randomized, KW09, braverman2013tight, ST13, SBPcc, OW15, Gav2016,braverman2017rounds,braverman2018near,kamath2021simple, DOR21}. 
Building upon communication lower bounds for set-disjointness, applications in diverse areas have been studied. For example, it gives lower bounds for monotone circuit depth \cite{goos2018communication}, streaming problems \cite{alon1999space,bar2004information,kamath2021simple}, proof complexity \cite{goos2018communication}, game theory \cite{stablemarriage,ganor2017communication}, property testing \cite{BBM11}, data structure lower bound \cite{MNSW98}, extension complexity \cite{braverman2013information,SBPcc}, and more.

Given the importance of this problem, many techniques were invented simply to understand communication lower bounds of set-disjointness. 
Some remarkable methods include rank method \cite{grigoriev1985lower,haastad2007randomized,rao2020communication}, discrepancy method \cite{rao2015simplified}, corruption bound \cite{razborov1992on}, smooth rectangle bound \cite{jain2010partition,HJ13}, and information complexity \cite{CSWY01,bar2004information, Gro, Jay}. 
Among all of these methods, the information complexity framework seems to provide the best results so far. We refer interested readers to \cite{chattopadhyay2010story} for a good survey on these results.

In this paper, we continue the study of set-disjointness. 
Inspired by \textit{simulation methods} in \textit{query-to-communication lifting theorems} \cite{raz1997separation,goos2018deterministic, lovett2022lifting, Yang2022SimulationMI}, we present a proof of lower bounds of set-disjointness based on \textit{density increment arguments} (sometimes also called structure-vs-pseudorandomness approach). 
Based on this method, we give several new lower bounds for set-disjointness in different communication models. 
Our proof can be considered as a combination of simulation methods and information complexity. 

Compared with previous techniques, our proof is simpler and more general. It addresses some drawbacks of both simulation methods and information complexity methods. More details will be discussed in Section \ref{sec: techniques}.

\subsection{Our results}

The main contribution of this work is ''\textit{explicit proofs}'' of communication lower bounds, together with some new \textit{unique-disjointness} lower bounds.
We call it explicit because our proof framework has several advantages compared to existing techniques:
\begin{itemize}
    \item It has fewer restrictions to communication models.
    \item It allows us to use communication lower bound techniques in a non-black-box way.
    \item It provides a method to analyze distributions with correlations between different coordinates.
\end{itemize}
In Section \ref{sec: potential}, we discuss three potential applications of these advantages. Each of them corresponds to an advantage here.

Our proof builds on a combination of simulation techniques from lifting theorems and information complexity. Specifically, we abstract the core idea from \textit{Raz-McKenzie simulation} \cite{raz1997separation} and revise it as a density increment argument. To explain more connections and comparison with previous techniques, we present three lower bounds for unique-disjointness problem.

We first study the \textit{multi-party communication model} ($\kdisj$). 
In this setting, there are $k$ parties, where each party $j$ holds a set $x_j\in\{0,1\}^{n}$ (we use a binary string to represent a set). 
It is promised that either all sets are pairwise disjoint, or they share a unique common element. Formally, we define 
\begin{itemize}
\item $D_{0}:=\{(x_1,\dots,x_k)\in(\{0,1\}^{n})^{k}: \forall i, x_{1}(i) + \dots + x_{k}(i)\leq 1 \}$.
\item $D_{*}:=\{(x_1,\dots,x_k)\in(\{0,1\}^{n})^{k}: \exists \ell, x_{1}(\ell) = \cdots =x_{k}(\ell)=1\text{ and } \forall i\neq \ell, x_{1}(i) + \dots + x_{k}(i)\leq 1 \}$.
\end{itemize}
We use $D_0$ to refer to the \textsf{no} instances and $D_*$ to refer to the \textsf{yes} instances. 
In this setting, we prove a structure lemma that any $0$-large rectangle must intersect $D_{*}$.

\begin{theorem}
\label{thm: kdisk_dcc}
Let $R\subseteq(\{0,1\}^{n})^{k}$ be a rectangle such that $|R\cap \Noinstance|\geq 2^{-n/k}\cdot |\Noinstance|$, then $R\cap\Yesinstance\neq\emptyset$.
\end{theorem}

We note that Theorem \ref{thm: kdisk_dcc} implies (and stronger than) an $\Omega(n/k)$ deterministic communication lower bound of $\kdisj$. 
For any protocol with $o(n/k)$ communication bits, we can always find a rectangle $R$ in the partition such that $|R\cap \Noinstance|\geq 2^{-n/k}\cdot |\Noinstance|$, Theorem \ref{thm: kdisk_dcc} then tells us that $R$ is not disjoint from $\Yesinstance$.

Our proof is a two-page elementary (and self-contained) proof. Furthermore, we do not even need notions like entropy or rank. 
This proof also reveals the main idea of query-to-communication lifting theorems.  We will discuss more details in Section \ref{sec: techniques}.

Our second contribution is a new proof of randomized communication lower bounds of $\kdisj$. This problem has been extensively studied for many years.
Building on a series of great papers \cite{alon1999space,bar2004information, CKX03}, the optimal tight randomized communication lower bound $\Omega(n/k)$ was finally obtained by \cite{Gro, Jay} through the information complexity framework. In this paper, we reprove this theorem via the density increment argument.

\begin{theorem}
\label{thm: kdisk_rcc}
For any $k\geq 2$, the randomized communication complexity of $\kdisj$ is $\Omega(n/k)$.
\end{theorem}

We first note that Theorem \ref{thm: kdisk_rcc} does not imply Theorem \ref{thm: kdisk_dcc} because Theorem \ref{thm: kdisk_dcc} indicates that every large rectangle (contains many \textsf{no} instance) cannot be monochromatic. However, Theorem \ref{thm: kdisk_rcc} only proves randomized communication lower bounds. 

Our proof of Theorem \ref{thm: kdisk_rcc} is a mix of information complexity and query-to-communication simulations. 
Roughly speaking, in the information complexity framework, we analyze the information cost for each coordinate and then apply a direct-sum argument to merge them. 
In our density increment argument, we merge these costs by borrowing the projection operation from query-to-communication simulations. 
Hence, our density increment argument can be interpreted as an alternative direct-sum argument. 

Several papers \cite{chattopadhyay2019query, IC0isIC,manor2022lifting} pointed research directions in connecting information complexity and lifting theorems, and our proof has a great potential to unify information complexity and lifting theorems in this direction. 

Our last result is a tight deterministic lower bound for (two-party) \textit{sparse unique-disjointess} ($\sparsedisj$) for a large range of sparse parameters. 
This problem, with sparsity parameter $s$, can be described as follows: Alice holds a set $A$ and Bob holds a set $B$ with $|A|,|B|\le s$. It is promised that either $A\cap B= \emptyset$ or $|A\cap B|=1$, where Alice and Bob need to distinguish the two cases with deterministic communication.

Two extreme choices of $s$ correspond to two important problems in communication complexity. 
If $s=n$, this problem becomes the standard unique-disjointness problem (i.e., $\kdisj$ with $k=2$). 
When $s=1$, the problem is essentially the \textit{EQUALITY problem}. 
For $\sparsedisj$, we prove the following theorem.

\begin{theorem}
\label{thm: sdisk_dcc}
Let $\epsilon>0$ be any small constant. For any $s\leq n^{1/2-\epsilon}$, the deterministic communication complexity of $s$-sparse unique-disjointness is $\Omega(s\cdot\log(n/s))$.
\end{theorem}

Prior to our work, Kushilevitz and Weinreb \cite{KW09} proved the same lower bound for a smaller range of $s\leq \frac{\log n}{\log \log n}$. Then Loff and Mukhopadhyay \cite{LMequality} improved this range to $s\leq n^{1/101}$. 
Our Theorem \ref{thm: sdisk_dcc} further pushes this range to $\approx n^{1/2}$. 

Our proof of Theorem \ref{thm: sdisk_dcc} is built on \cite{LMequality} with several differences where the main difference is that we no longer fully simulate the communication tree by a decision tree. Instead, we aim to find a long path in the communication tree. The approach was suggested by Yang and Zhang \cite{Yang2022SimulationMI}. We believe it is possible to further improve the range to all $s\geq 1$ and discuss more details in Section \ref{sec: sparse_DCC}.

A similar task is to prove a deterministic lower bound for sparse set-disjointness without the unique requirement. 
In this setting, H{\aa}stad and Wigderson pointed out \cite{haastad2007randomized} a same $\Omega(s\cdot\log(n/s))$ bound can be proved via the rank method in \cite{jukna2011extremal}. 
However, in the unique setting, \cite{KW09} showed that the rank method is impossible to achieve such tight bounds.

We emphasize that Theorem \ref{thm: sdisk_dcc} is a lower bound only for \textit{deterministic} communication complexity. 
Allowing public randomness and a constant error, there exists a protocol that costs only $O(s)$ bits \cite{haastad2007randomized}. 
Therefore, this $\log(n/s)$ factor is a \textit{separation} between randomized communication and deterministic communication. 
This also implies all lower bound techniques that simultaneously imply randomized communication lower bounds, including information complexity approaches, cannot reprove our bounds.

Furthermore, Braverman \cite{braverman2012interactive} gave a zero-error protocol for $1$-sparse unique-disjointness with constant information cost which can be extended to all $s\geq 1$. 
\begin{lemma}
For any $s>0$. There is a zero-error protocol for $s$-sparse unique-disjointness with information cost $O(s)$.
\end{lemma}
Overall, Theorem \ref{thm: sdisk_dcc} demonstrates that the density increment argument has fewer restrictions on communication models and is able to circumvent such barriers by rank methods and information complexity.

\subsection{Our techniques}
\label{sec: techniques}

Here we give an overview of our proof technique and discuss connections to lifting theorems and information complexity. 
We focus on Theorem \ref{thm: kdisk_dcc}. 
We recall that \textsf{no} instances are
\[
D_{0}:=\{(x_1,\dots,x_k)\in(\{0,1\}^{n})^{k}: \forall i, x_{1}(i) + \dots + x_{k}(i)\leq 1 \}.
\]
Our main idea of Theorem \ref{thm: kdisk_dcc} is a density increment argument. In this argument, we first define the density of a rectangle $R$ on $\Noinstance$ by
\[
E(R):=\frac{|R\cap \Noinstance|}{|\Noinstance|}=\frac{|R\cap \Noinstance|}{(k+1)^{n}}.
\]
It is clear that $E(R)\leq 1$ because $R\cap \Noinstance\subseteq \Noinstance$. 
Now Theorem \ref{thm: kdisk_dcc} is equivalent to say that any rectangle $R$ with density $E(R)\geq 2^{-o(n/k)}$ cannot be monochromatic.

Let $R$ be any \textit{monochromatic rectangle} containing only \textsf{no} instances. 
We will perform a projection operation to convert $R$ into another \textit{monochromatic rectangle} $R'\subseteq (\{0,1\}^{n-1})^{k}$ with a larger density $E(R')\geq E(R)\cdot (1+1/k)$. 
Now, since $R'$ is still monochromatic, we repeat this projection for $n$ rounds where each round increases the density by a factor of $(1+1/k)$. 
Let $R^*$ be the rectangle after $n$ projections, we then have
\[
E(R^{*})\ge E(R)\cdot (1+1/k)^{n}.
\]
Combining $E(R^*)\le1$ and $E(R)\ge2^{-o(n/k)}$, this gives a contradiction.

Now we briefly explain our projection process. 
Let $R = X_1\times \dots\times X_{k}\subseteq (\{0,1\}^{n})^{k}$ be a monochromatic rectangle. For each party $j\in[k]$, the projection of $R$ on $j$ is a rectangle $\Pi_{j}(R) =X_{1}'\times\cdots\times X_{k}'\subseteq(\{0,1\}^{n-1})^{k}$ defined by:
\begin{itemize}
\item For each party $j'\neq j$, $X_{j'}':=\{x'\in \{0,1\}^{n-1}:  (x',0)\in X_{j'}\}$.
\item For the party $j$, $X_{j}':=\{x'\in \{0,1\}^{n-1}:\text{ either }(x',0)\in X_{j}\text{ or } (x',1)\in X_{j}\}$.
\end{itemize}
It is not hard to see that $\Pi_{j}(R)$ (for any $j$) preserves the monochromatic property of $R$. 
On the other hand, we show there exists a party $j\in[k]$, such that $\Pi_{j}(R)$ has a larger density compared to $R$. In fact, this density increment captures the communication cost of party $j$. 

We give a full proof of Theorem \ref{thm: kdisk_dcc} in Section \ref{sec: weak_dcc}. We suggest readers begin with Section \ref{sec: weak_dcc}, and then proceed to Section \ref{sec: rcc_vs_udisj} and Section \ref{sec: sparse_DCC}.

\paragraph{Connections to lifting theorems.}
Query-to-communication lifting theorems are a generic method to lift the hardness of one-party functions (decision tree lower bounds) to two-party lower bounds in the communication model. 
This recent breakthrough gives many applications in diverse areas \cite{clique_independent,huynh2012virtue,goos2018deterministic,goos2018extension,de2016limited,de2019lower,ggks18,CKLM18,GR18,Nashpotentialgame,robere2016exponential,pitassi2017strongly,lee2015lower,CLRS16,KMR17,approximatelogrank,goos2022lower}. The simulation method is widely used to prove such lifting theorems. 
During the simulation, we maintain some pseudorandom properties to enforce the communication protocol to be similar to a query protocol. In this process, a potential function is always used to capture the number of communication bits against the pseudorandom property.

In this paper, we adopt the idea of the potential function argument and rephrase it as a density increment argument. In lifting theorems, the simulation is mainly used to maintain a good structure of rectangles, such as maintaining full-range rectangles. In contrast, our density increment argument is more flexible.

\paragraph{Connections to information complexity.} 
The information complexity is another important tool to prove communication complexity lower bounds. It aims to analyze the mutual information between communication transcripts and the random inputs held by Alice and Bob. From information-theoretical perspective, the randomized communication complexity is then lower bounded by the mutual information. To analyze the mutual entropy, a very useful tool here is called the direct-sum argument \cite{CSWY01,bar2004information,CKX03,Gro,Jay}. Roughly speaking, based on this argument, we only need to analyze the mutual information of the communication transcripts and each coordinate of Alice and Bob's input. This step significantly reduces the difficulty of analyzing mutual entropy. 

In our proof (see Section \ref{sec: rcc_vs_udisj} for example), we also use the local-to-global strategy. The difference is that the information complexity paradigm \cite{bar2004information,CKX03,Gro,Jay} uses a direct-sum argument, which is a useful tool in information theory. In comparison, the density increment arguments use a combinatorial operation called projection. The projection provides more flexibility in different applications.  As \cite{DOR21} pointed out many applications in other settings are not amenable to the standard direct sum argument, such as proving information-theoretic lower bounds for the number-on-forehead model. In the density increment arguments, we do not see such barriers for now.

\subsection{Potential applications of explicit proofs}
\label{sec: potential}

We discuss some potential applications of our density increment arguments.

\paragraph{Different communication models.}
As an active research area, many techniques were invented to prove communication complexity lower bounds in past decades. However, many of these techniques are specific to only one communication model. 
For example, the rank method mainly applies to deterministic communication; information complexity is usually used for randomized communication. We believe our density increment arguments provide more flexibility, having less dependency on the communication models. For example, in Theorem \ref{thm: kdisk_dcc}, we study a lower bound similar to (but not exactly) corruption bounds; in Theorem \ref{thm: kdisk_rcc}, we prove randomized communication lower bounds; in Theorem \ref{thm: sdisk_dcc}, we show deterministic lower bounds (it is a separation from randomized communication).
Overall, we demonstrate that (at least for the set-disjointness problems) the density increment arguments combine both advantages of lifting theorems and information complexity. 
Till now, some communication models are still not fully understood. For example, the $(\exists -1)$-game is an interesting communication model with applications in extension complexity \cite{goos2018extension}. 
However, to our best knowledge, we still do not have a generic way to prove extension complexity lower bounds. 
Another example is the number-on-forehead model. It would be interesting to see if density increment arguments give new applications in various communication models.

\paragraph{Streaming lower bounds.}
The connections between communication complexity and streaming lower bounds were explored by a seminal work of Alon, Matias and Szegedy \cite{alon1999space}. 
\cite{alon1999space} proved a streaming lower bound for \textit{frequency moment estimation} based on a reduction from unique-disjointness lower bounds. 
After that, many subsequent works made a lot of efforts to improve the lower bounds \cite{bar2004information,CKX03,Gro,chakrabarti2008robust,andoni2008better,Huang09} for this problem. 
As \cite{crouch2016stochastic} also pointed out, any improved lower bounds of frequency moment estimation can be automatically applied to improve lower bounds for many other streaming problems.

However, the optimal bound for this fundamental problem\footnote{We focus on the random order streaming model.} is still not clear\footnote{\cite{Huang09} claimed a tight bound, however, \cite{crouch2016stochastic} pointed there is a flaw in \cite{Huang09}.}. 
To the best of our knowledge, all current lower bounds rely on (black-box) reductions from $\kdisj$ lower bounds. 
As we discussed, an $\Theta(n/k)$ bound for randomized communication of $\kdisj$ is already tight, and the black-box reduction seems a dead end in achieving tight bounds for frequency moment estimation.

To resolve this barrier, we believe an important step is to open the black box. Put differently, we should extend communication complexity lower-bounds techniques into streaming models. Since our proof of $\kdisj$ has fewer restrictions to models, it is reasonable to try this argument on streaming settings. Concretely, could we prove a tight lower bound of frequency moment estimation by the density increment argument?

\paragraph{Coordinate-wise correlated hard distributions.} 
Many proofs of randomized communication complexity lower bounds start with Yao's minimax theorem and design a hard distribution.
In some important applications, the hard distribution has a strong correlation between input coordinates. 
A good example is \textit{Tseitin problems}, the lower bounds of which can be converted into lower bounds in proof complexity  \cite{goos2018communication}, extension complexity \cite{goos2018extension}, monotone 
computation lower bounds \cite{pitassi2017strongly}. 
However, the hard distribution of Tseitin has complicated coordinate-wise correlation, which makes the information complexity argument difficult to use. 
Known lower bounds for randomized communication \cite{goos2018communication, goos2018extension} all lose a $\log n$ factor (including one based on a black-box reduction from two-party unique-disjointness \cite{goos2018communication}). Again, it seems the loss cannot be avoided in black-box reductions and it is very interesting to see if our density increment arguments are able to break this barrier.

\paragraph{Acknowledgements.} Authors thank Shachar Lovett and Xinyu Mao for helpful discussions. We are grateful to Kewen Wu for reading the early versions of this paper and providing useful suggestions.

\section{Preliminary}\label{sec: preliminary}

For integer $n\ge0$, we use $[n]$ to denote the set $\{1,2,\ldots,n\}$.
Throughout, $\log(\cdot)$ is the logarithm with base $2$.
For a finite domain $X$, we use $x\sim X$ to denote a random variable $x$ uniformly distributed over $X$. 

\begin{definition}[Entropy]
Let $D$ be a random variable on $X$. The entropy of $D$ is defined by
\[
\ent(D):=\sum_{x}\Pr[D=x]\cdot\log(1/\Pr[D=x]).
\]

Let $A$ and $B$ be two random variables on $X$ and $Y$ respectively. the conditional entropy of $A$ given $B$ is defined by
\[
\ent(A|B) = \sum_{y\in Y}\Pr[B=y] \cdot \sum_{x\in X}\Pr[A=x|B=y]\cdot\log(1/\Pr[A=x|B=y]).
\]
\end{definition}

\begin{definition}[Mutual information]
Let $A$ and $B$ be two (possibly correlated) random variables on $X$ and $Y$ respectively.
The mutual information of $A$ and $B$ is defined by 
\[
\mutualent(A:B) = \ent(A) - \ent(A|B) = \ent(B) - \ent(B|A).
\]

Let $C$ be a random variable on $Z$, the  conditional mutual information of $A$ and $B$ given $C$ is defined by
\[
\mutualent(A:B|C) = \ent(A|C) - \ent(A|B,C) = \ent(B|C) - \ent(B|A,C).
\]
\end{definition}

We use several basic properties of entropy and mutual information.

\begin{fact}\label{fact: entropy}
Let $A$ and $B$ be two (possibly correlated) random variables on $X$ and $Y$ respectively.
\begin{enumerate}
     \item Conditional entropy inequality: $\ent(B|A) \leq \ent(B)$.
    \item Chain rule: $\ent(A,B) = \ent(B) + \ent(B|A) =  \ent(A) + \ent(A|B)$.
    \item Nonnegativity: $\mutualent(A:B)\geq 0$.
    \item $\mutualent(A:B)\leq\max\{\ent(A),\ent(B)\}$. 
\end{enumerate}
\end{fact}
\section{Deterministic lower bound for multi-party unique-disjointness}
\label{sec: weak_dcc}
In this section, we give a simple proof for Theorem \ref{thm: kdisk_dcc}. Our proof is based on a density increment argument. We first formally define the problem. In our writing, we use binary strings to represent a set. 
We associate any set $A\subseteq[n]$ with a corresponding string $x\in\{0,1\}^{n}$ by setting $x(i)=1$ iff $i\in A$.

\begin{definition}[$\kdisj$, deterministic version]
\label{def: kdisj}
For each $k\geq 2$ and $n\geq 1$. 
We define $D_{0}$ (\textsf{no} instances) and $D_{*}$ (\textsf{yes} instances) as follows:
\begin{itemize}
    \item $D_{0}:=\{(x_1,\dots,x_k)\in(\{0,1\}^{n})^{k}: \forall i, x_{1}(i) + \dots + x_{k}(i)\leq 1 \}$.
    \item $D_{*}:=\{(x_1,\dots,x_k)\in(\{0,1\}^{n})^{k}:\exists \ell, x_{1}(\ell) = \cdots =x_{k}(\ell)=1\text{ and } \forall i\neq \ell, x_{1}(i) + \dots + x_{k}(i)\leq 1 \}$.
\end{itemize}
A $k$-party deterministic communication protocol $C$ solves $\kdisj$ if,
\begin{itemize}
    \item for all $(x_1,\dots,x_{k})\in D_0$, $C(x_1,\dots,x_k)=0$,
    \item for all $(x_1,\dots,x_{k})\in D_*$, $C(x_1,\dots,x_k)=1$.
\end{itemize}
\end{definition}

\noindent Since the projection may fix some coordinates, we also define the projected instances. For a set $I\subseteq[n]$, we define $\Noinstance^{I}$ (\textsf{no} instances on $I$) as
\[
D_{0}^{I}:=\{(x_1,\dots,x_k)\in(\{0,1\}^{I})^{k}: \forall i\in I, x_{1}(i) + \dots + x_{k}(i)\leq 1 \},
\]
and define $\Yesinstance^I$ (\textsf{yes} instances on $I$) as
\[
D_{*}^I:=\{(x_1,\dots,x_k)\in(\{0,1\}^{I})^{k}:\exists i\in I, x_{1}(i) = \cdots =x_{k}(i)=1\text{ and } \forall i'\neq i, x_{1}(i') + \dots + x_{k}(i')\leq 1 \}.
\]
We also partition the \textsf{yes} instances by $D_{*}^{I}= \bigcup_{i\in I} D_{i}^{I}$, where
\[
D_{i}^I:=\{(x_1,\dots,x_k)\in(\{0,1\}^{I})^{k}: x_{1}(i) = \cdots =x_{k}(i)=1\text{ and } \forall i'\neq i, x_{1}(i') + \dots + x_{k}(i')\leq 1 \}.
\]

Now we define the density function.

\begin{definition}[Density function]
For each $I\subseteq[n]$ and $R = X_1\times \dots\times X_{k}\subseteq (\{0,1\}^{I})^{k}$, we define its density function as 
\[
E^{I}(R):=\log\left(\frac{|R\cap D^I_{0}|}{(k+1)^{|I|}}\right).
\]
\end{definition}

\noindent Note that $E^{I}(R)\leq 0$ for any rectangle $R$ because $|D_0^{I}| = (k+1)^{|I|}$. We simplify the notation as $E(R)$ if $I$ is clear in the context. A crucial step in our argument is the projection operation. 
\begin{definition}[Projection]
Let $R = X_1\times \dots\times X_{k}\subseteq (\{0,1\}^{I})^{k}$ be a rectangle. For an $i\in I$ and $j\in[k]$, the projection of $R$ on $(i,j)$ is a rectangle $\Pi_{i,j}(R) =X_{1}'\times\cdots\times X_{k}'\subseteq(\{0,1\}^{I\setminus\{i\}})^{k}$ defined by:
\begin{itemize}
    \item for each $j'\neq j$, $X_{j'}':=\{x'\in \{0,1\}^{I\setminus \{i\}}:  (x',0)\in X_{j'}\}$,
    \item for $j$, $X_{j}':=\{x'\in \{0,1\}^{I\setminus \{i\}}:\text{ either }(x',0)\in X_{j}\text{ or } (x',1)\in X_{j}\}$.
\end{itemize}
Here $(x',0)$ is a string in $\{0,1\}^{I}$ by extending $x'\in\{0,1\}^{I\setminus\{i\}}$ with $x_{i}=0$.
\end{definition}

The projection operation has two useful properties. The first one is that projection preserves the monochromatic property of the rectangle.

\begin{fact}
\label{fact: k_u_m}
Let $R$ be a rectangle such that $R\cap D_{*}^{I}=\emptyset$. Then for every $i\in I$ and $j\in[k]$, we have 
\[
\Pi_{i,j}(R) \cap D_{*}^{I\setminus\{i\}}=\emptyset.
\]
\end{fact}
The proof of Fact \ref{fact: k_u_m} follows from the definition and we omit it here. The next property is phrased as the following projection lemma.
\begin{lemma}[Projection lemma]
\label{lemma: k_u_p}
Let $R = X_1\times \dots\times X_{k}\subseteq (\{0,1\}^{I})^{k}$ be a rectangle. If there is a coordinate $i\in I$ such that $R\cap D_{i}^{I} = \emptyset$, then there is some $j\in[k]$ such that
\[
E^{I\setminus \{i\}}(\Pi_{i,j}(R)) \geq E^{I}(R) + 1/k.
\]
\end{lemma}

Given Lemma \ref{lemma: k_u_p} and Fact \ref{fact: k_u_m}, Theorem \ref{thm: kdisk_dcc} becomes straightforward. We can simply repeat the projection $n$ times for $i\in[n]$, where each time we use Lemma \ref{lemma: k_u_p} to choose a good $j$ to do projection and increase the density function by $1/k$.
Now we prove Lemma \ref{lemma: k_u_p}.

\begin{proof}[Proof of Lemma \ref{lemma: k_u_p}]
Let $R$ be a rectangle such that $R\cap D_{i}^{I} =\emptyset$. 
Let $I':=I\setminus\{i\}$ and 
\[
L:= \{x'\in D_{0}^{I'}: \exists x\in R\cap D_0^{I},~ x|_{I'} = x'\}.
\]
Here $x|_{I'}$ is the restriction of $x$ on $I\setminus\{i\}$. 
Note that for all $j\in [k]$, $\Pi_{i,j}(R)\cap D_{0}^{I'}=\Pi_{i,j}(R)\cap L$, and our goal is to show that there is a $j$ such that $|\Pi_{i,j}(R)\cap L|$ is large.

For every $x'\in L$, define the extension set of $x'$ as $\ext(x'):=\{x\in R\cap D_{0}^{I}: x|_{I'} = x'\}$.
Crucially, for every $x'\in L$, we have
\begin{align}\label{eq:ext_k}
|\ext(x')|\leq k.
\end{align}
Note that, without the condition $R\cap D_{i}^{I}=\emptyset$, it can only be bounded by $k+1$.
Inequality (\ref{eq:ext_k}) is proved by contradiction. 
If there is a $x'=(x_1',\dots,x_k')\in L$ such that $|\ext(x')|= k+1$.
Then we must have that $(x_1',1)\in X_1,\dots,(x_k',1)\in X_k$, contradicting $R\cap D_{i}^{I}=\emptyset$. 

We now continue our proof. 
Partition $L$ into two parts:
\[
A:=\{x'\in L: |\ext(x')|\geq 2 \}
\quad\text{and}\quad
B:=\{x'\in L: |\ext(x')|=1 \}.
\]
First observe that for any $x'=(x_1',\ldots,x_k')\in A$, we have $(x_j',0)\in X_j$ for every $j\in[k]$ as $R$ is a rectangle. 
This implies $x' \in \Pi_{i,j}(R)$ for all $j\in [k]$. Hence
\[
|A| = |A \cap \Pi_{i,j}(R)|,\quad\forall j\in[k].
\]
Applying (\ref{eq:ext_k}) with $x'\in A$, we have
\[
k\cdot |A \cap \Pi_{i,j}(R)|= k\cdot |A| \geq |\{x\in R\cap D_{0}^{I}: x|_{I'} \in A\}| ,\quad\forall j\in[k].
\]
For $x'\in B$, since $|\ext(x')|=1$, we have
\[
|\{x\in R\cap D_{0}^{I} :x|_{I'}\in B\}| = |B|.
\]
On the other hand, for every $x'\in L$, it always exists one $j\in [k]$ such that $x' \in \Pi_{i,j}(R)$. 
By an average argument, there is at least one $j\in [k]$ such that, 
\[
k\cdot |B \cap \Pi_{i,j}(R)|\geq |B| = |\{x\in R\cap D_{0}^{I} :x|_{I'}\in B\}|.
\]
As a result, for this fixed $j$ we have
\begin{align*}
k\cdot |L\cap \Pi_{i,j}(R)| 
&= k\cdot |B \cap \Pi_{i,j}(R)| + k\cdot|A\cap \Pi_{i,j}(R)|\\
&\geq |\{x\in R\cap D_{0}^{I}: x' \in A\}| + |\{x\in R\cap D_{0}^{I} :x' \in B\}| \\
&= |R\cap D_{0}^{I}|.
\end{align*}
By the definition of the density function, we have
\begin{align*}
E^{I'}(\Pi_{i,j}(R))
&=\log\left(\frac{|\Pi_{i,j}(R)\cap D_0^{I'}|}{(k+1)^{|I'|}}\right)
=\log\left(\frac{(k+1)\cdot|\Pi_{i,j}(R)\cap L|}{(k+1)^{|I|}}\right)\\
&\ge\log\left(\frac{(k+1)\cdot|R\cap D_0^I|}{k\cdot(k+1)^{|I|}}\right)
=E^I(R) + \log (1+1/k)\\
&\geq E^I(R) + 1/k.
\tag{since $\log(1+1/k)\ge1/k$ for all $k\ge1$}
\end{align*}
\end{proof}

\section{Randomized lower bound for multi-party unique-disjointness}
\label{sec: rcc_vs_udisj}

We focus on randomized communication lower bounds in this section. By Yao's minimax theorem, this is equivalent to identifying a distribution $\calP$ that is hard on average for any deterministic communication protocol. 
We use the notation $D_0,D_{*},D_{0}^{I},D_{*}^{I}$ same as the previous section (Definition \ref{def: kdisj}). 

Our hard distribution $\calP$ is supported on $D_0\cup D_*$.
\begin{definition}\label{hard_distribution}
For any $n,k\geq 1$, we define the hard distribution $\calP$ on $(\{0,1\}^{n})^{k}$ as follows.
\begin{framed}
\begin{enumerate}
    \item For every $i\in[n]$, uniformly and independently sample $\calW_{i}\sim [k]$ and $\calA_{i}\sim\{0,1\}$.
    \item For every $i\in[n],j\in[k]$, if $\calW_{i}=j$ and $\calA_{i}=1$, then set $x_{j}(i)=1$; otherwise set $x_j(i)=0$.
    \item Sample $\calB\sim \{0,1\}$ and $\mathcal{\ell}\sim[n]$ uniformly. If $\calB=1$, then update $x_{j}(\ell)=1$ for all $j\in [k]$.
    \item Output $x=(x_1,\dots,x_k)$.
\end{enumerate}
\end{framed}
Given this hard distribution $\calP$, we also define the distribution $\calQ:=(\calP\mid \calB=0)$.
\end{definition}
\noindent Now we give some explanation of the random variables in this sampling process. 
\begin{itemize}
    \item The bit $\calB$ determines whether to output \textsf{yes} instances or \textsf{no} instances.
    In particular, for $\calB=1$, we output a \textsf{yes} instance. Hence we update $x_{j}(\ell)=1$ for all $j\in [k]$, where $\ell$ is uniformly sampled.
    \item For every $i$, the variable $\calW_i\in [k]$ captures which party ($\calW_i$) may have the $i$-th element
    \item For every $i$, $\calA_i$ determines whether $\calW_i$ has the $i$-th element or not.
\end{itemize}

It is well-known that a deterministic protocol $C$ with communication complexity $c$ partitions the input domain into at most $2^{c}$ rectangles where each rectangle corresponds to a leaf in the communication tree. 
We then define the following random variable $\calR$, which is the rectangle of a random leaf induced by the input distribution $\calQ$.

\begin{definition}
For a fixed deterministic protocol $C$, we define a distribution $\calR_{C}$ on leaf rectangles (of $C$) as follows.
\begin{enumerate}
    \item Randomly sample $x\sim \calQ$.
    \item Output the rectangle $R$ of $C$ containing $x$.
\end{enumerate}
\end{definition}

We emphasize $\calR$ is defined by $\calQ$, which is \textit{not} $\calP$. Hence for a protocol $C$ with a small error on $\calP$, the random rectangle $R\sim\calR$ should be biased towards $\Noinstance$ with high probability.

If $C$ is clear in the context, we will simply use $\calR$ for $\calR_C$. 
For any rectangle $R$, we also use $\calQ_{R}$ to denote the distribution $(\calQ\mid \calR=R).$
Let $\calW=(\calW_1,\dots,\calW_n)$ and we use $(\calQ,\calW)$ to denote the joint distribution of $\calQ$ and $\calW$. 
We are now ready to explain our theorems.
\begin{theorem}\label{ic_cost}
Let $0<\epsilon<0.0001$ be a constant. For any deterministic protocol $C$ with error $\epsilon$ under $\calP$, i.e., 
\[
\Pr_{x\sim \calP}[C(x) = \kdisj(x)] \geq 1-\epsilon.
\]
We have
\[
\mutualent(\calQ,\calW:\calR_{C}) = \Omega(n/k).
\]
\end{theorem}

We note that Theorem \ref{ic_cost} implies Theorem \ref{thm: kdisk_rcc} because the communication complexity of $C$ is lower bounded by $\ent(\calR)$, and $\ent(\calR)$ is an upper bound of $\mutualent((\calQ,\calW):\calR)$ (Fact \ref{fact: entropy}).

A similar lower bound $\mutualent(\calQ:\mathcal{R}|\calW)$ was previously obtained by the information complexity framework \cite{Gro}. We reprove it by a density increment argument. 
In what follows, we fix the protocol $C$. We first give a high-level view of our proof. 

\paragraph{Sketch of the proof.} 
We first reinterpret the proof of the deterministic lower bound (Section \ref{sec: weak_dcc}) in an entropy perspective. Then we generalize it to randomized communication lower bounds.

Let $C$ be a deterministic communication protocol for $\kdisj$. Every leaf of $C$ is a monochromatic rectangle. Let $R$ be any $0$-monochromatic rectangle (i.e., $R\cap D_{*}=\emptyset$) of $C$. Then for every input $x^{*}\in R\cap D_0$ and $i\in[n]$, 
\begin{align}
\label{inequality: 2}
\Pr_{x=(x_1\dots,x_k)\sim \calP}\left[ x_1(i)=\cdots=x_k(i)=1\mid x\in R\text{ and }\forall i'\neq i, x(i') = x^{*}(i') \right] = 0
\end{align}

since $R$ is $0$-monochromatic. Furthermore, since $R$ is a rectangle, there is a party $j$ such that 
\[
\Pr_{x=(x_1\dots,x_k)\sim \calP}\left[ x_j(i)=1\mid x\in R\text{ and }\forall i'\neq i, x(i') = x^{*}(i') \right] = 0.
\]

Recall $\calQ=(\calP\mid \calB=0)$ samples \textsf{no} instances. Thus we also have
\[
\Pr_{x=(x_1\dots,x_k)\sim \calQ}\left[ x_j(i)=1\mid x\in R\text{ and }\forall i'\neq i, x(i') = x^{*}(i') \right] = 0.
\]
By the definition of $\calR$, it is equivalent to \footnote{In the following part, we use replace the notation $x\in R$ with $\calR=R$ when $x\sim \calQ$.}
\[
\Pr_{x=(x_1\dots,x_k)\sim \calQ}\left[ x_j(i)=1\mid \calR=R\text{ and }\forall i'\neq i, x(i') = x^{*}(i') \right] = 0.
\]

If we use the entropy language and recall the definition of $\calW = (\calW_1,...,\calW_n)$, it is equivalent to
\[
\ent\left(x_{j}(i) \mid \calR=R, \calW_i=j\text{ and } \forall i'\neq i, x(i') = x^{*}(i')\right) = 0.
\]
In contrast, if we do not condition on $R$, we have that 
\[
\Pr_{x=(x_1\dots,x_k)\sim \calQ}\left[ x_j(i)=1\mid \calW_i=j\text{ and }\forall i'\neq i, x(i') = x^{*}(i') \right] = 1/2.
\]
which can be written as,
\[
\ent\left(x_{j}(i) \mid \calW_i=j\text{ and } \forall i'\neq i, x(i') = x^*(i')\right) = 1.
\]
This gap captures the mutual information of $\calR$ and $(\calQ,\calW)$ on the $i$-th coordinate.

For different choices of $x^*\in R\cap D_0$, we may have different $j$ witnessing the mutual information. But on average, we have 
\[
\mathbb{E}_{j}\left[\ent\left(x_{j}(i) \mid \calR=R, \calW_i=j\text{ and } x(1),\dots,x(i-1),x(i+1),\dots,x(n) \right) \right]\leq 1 - 1/k.
\]
In particular, there exists a $j\in[k]$ such that 
\[
\ent\left(x_{j}(i) \mid \calR=R, \calW_i=j, x(1),\dots,x(i-1),x(i+1),\dots,x(n) \right) \leq 1 - 1/k.
\]
Now we explain how can we view the projection as a decoupling process for this mutual information. We can decompose the projection in two steps:
\begin{framed}
\begin{enumerate}
    \item Fix $\calW_{i}=j$, i.e., update $(\tilde{\calQ},\tilde{\calW})\gets (\calQ,\calW\mid \calW_i=j)$.
    \item Update the density function as 
    \[
    \ent\left(\tilde{\calQ}_{[n]\setminus\{i\}}, \tilde{\calW}_{[n]\setminus\{i\}} \mid \calR=R\right) - \ent\left(\tilde{\calQ}_{[n]\setminus\{i\}}, \tilde{\calW}_{[n]\setminus\{i\}}\right),
    \]
    or equivalently
    \[
    \ent\left(\calQ_{[n]\setminus\{i\}}, \calW_{[n]\setminus\{i\}} \mid \calR=R, \calW_i=j\right) - \ent\left(\calQ_{[n]\setminus\{i\}}, \calW_{[n]\setminus\{i\}}\mid \calW_i=j\right),
    \]
    where $\calQ_{[n]\setminus\{i\}}$ (resp., $\tilde\calQ_{[n]\setminus\{i\}}$) is the marginal distribution of $\calQ$ (resp., $\tilde\calQ$) on $[n]\setminus \{i\}$.
\end{enumerate}
\end{framed}
In the first step, we pick the party $j$ that has mutual information.
In the second step, we decouple the mutual information by simply removing it from the density function. 
The projection lemma (Lemma \ref{lemma: k_u_p}) captures how this decoupling step increases the density function. 
Another crucial fact is that, for any $0$-monochromatic rectangle $R$, the distribution $(\calQ_{[n]\setminus\{i\}} \mid \calR=R, \calW_i=j)$ is also supported on $D_{0}^{[n]\setminus\{i\}}$ (see Fact \ref{fact: k_u_m}), which guarantees us to continuously increase the density by projections on different coordinates. 

Now we generalize this to the randomized communication setting where the rectangle $R$ is not necessarily monochromatic. 
By the correctness of the protocol, most rectangle $R$ is biased to either \textsf{yes} instances or \textsf{no} instances. 

For a rectangle $R$ biased to \textsf{no} instances, we expect an inequality similar to (\ref{inequality: 2}) will hold: 
For most $R\sim \calR$, most \textsf{no} instance $x^*\sim (\calQ\mid \calR=R)$ and most $i\sim[n]$, it satisfies
\[
\Pr_{x=(x_1\dots,x_k)\sim \calP}\left[ x_1(i)=\cdots=x_k(i)=1\mid x\in R\text{ and }\forall i'\neq i, x(i') = x^{*}(i') \right] \leq \delta,
\]
where $\delta$ is a small constant depending on the error rate $\epsilon$ of the protocol.

On the other hand, we also need to argue that projections can be repeated. This part is slightly more complicated than the deterministic case where we can simply fix $\calW_i=j$ for some $j\in[k]$. 
In the randomized case, we cannot fix it because we have to preserve the bias. 
This is addressed by:
\begin{itemize}
    \item Bias lemma (Lemma \ref{Bias_lemma}), a randomized variant of Fact \ref{fact: k_u_m}.
    \item Projection lemma (Lemma \ref{Projection_lemma}), a randomized variant of Lemma \ref{lemma: k_u_p}.
\end{itemize}

\subsection{Key definitions and lemmas}

Now we introduce the key definitions and lemmas (bias lemma and projection lemma) needed for the randomized communication lower bound.

\begin{definition}[$\rho$-restriction]
For $J\subset[n]$ and $w_J\in [k]^{J}$, we call $\rho=(J,w_J)$ an $\rho$-restriction, and denote $(\calQ,\calW|_\rho)$ as the distribution $(\calQ,\calW\mid \forall i\in J, \calW_i = w_i)$. 
\end{definition}

The $\rho$-restriction corresponds to projections in the deterministic case: For $\rho = (J,w_J)$, $i\in J$ corresponds to the projection $\Pi_{i,w_i}$. 
Now we define our new density function.

\begin{definition}[Density function]
Let $R$ be a rectangle and a set $I \subseteq [n]$. For a restriction $\rho = (I^{c}, w_{I^{c}})$ with $I^c=[n]\setminus I$, its density is defined by
\[
E^{I}(R,\rho) :=  \ent(\calQ_I,\calW_I \mid \rho, \calR=R) -\ent(\calQ_I,\calW_I\mid \rho).
\]
The average density is defined by
\[
E^{I}:= 
\underset{(\rho,R)\sim (\calW_{I^{c}},\calR)}{\mathbb{E}}\left[E^{I}(R,\rho)\right]
=\underset{\rho}{\mathbb{E}}\left[-\mutualent(\calQ_I,\calW_I:\calR\mid \rho)\right]
=-\mutualent(\calQ_I,\calW_I:\calR\mid \calW_{I^c}).
\]
In particular $E^{[n]}=-\mutualent(\calQ,\calW:\calR)$.
\end{definition}
The main difference between the deterministic setting and randomized setting is that, in the deterministic case, we consider $E^{I}(R,\rho)$ for some fixed $R$ and $\rho$. However, in the randomized communication case, we have to consider $E^{I}$ which does not fix $R$ and $\rho$ because the projection lemma (Lemma \ref{Projection_lemma}) and the bias lemma (Lemma \ref{Bias_lemma}) are not preserved under fixed $R$ and $\rho$.

We also note that $-E^{I}(R,\rho)$ might be negative for some $R$ and $\rho$. But $-E^{I}$ is always nonnegative because it is mutual information.

As we mentioned before, in the randomized setting, the leaves are no longer monochromatic but biased. 
Now we define the following bias definition to capture it (This is a randomized version of equation \ref{inequality: 2}).

\begin{definition}
Let $R$ be a rectangle and $I\subseteq[n]$. Let $\rho = (I^{c}, w_{I^{c}})$ be a restriction. For any $i\in I$ and input $x^{*}\in D_{0}^{I\setminus\{i\}}$, the bias of $x^{*}$ on the coordinate $i$ under $(R,\rho)$ is defined by 
\[
\gamma_{i,\rho,R}^{I}(x^{*}):=\Pr_{x=(x_1\dots,x_k)\sim \calP}\left[ x_1(i)=\cdots=x_k(i)=1 ~\middle\vert~ \rho, x\in R, x\notin \bigcup_{\ell\in I^{c}} D_{\ell} \text{ and }\forall i'\in I\setminus\{i\}, x(i') = x^{*}(i')\right],
\]
where $D_\ell\subseteq D_*$ is the \textsf{yes} instances with intersection witnessed by $x(\ell)$, i.e., $D_\ell$ is support of $(\calP|\calB=1,\ell)$.\footnote{Though $R$ is the leaf conditioned on an input from $\calQ=(\calP|\calB=0)$, it is still possible that $\calP(D_{\ell}\cap R)>0$ since the protocol is allowed to err. That is why $x\notin \bigcup_{\ell\in I^{c}} D_{\ell}$ is not implied by $\calR=R$.}
Then we define the average bias of a rectangle $R$ on $i$ as
\[
\gamma_{i,R}^{I}:= \underset{( x^{*},\rho)\sim (\calQ_{I\setminus\{i\}},  \calW_{I^{c}}\mid\calR=R)}{\mathbb{E}}\left[\gamma_{i,\rho,R}^{I}(x^{*})\right].
\]

The overall bias on $i$ is defined by
\[
\gamma_{i}^{I}:= \underset{R\sim\calR}{\mathbb{E}}\left[\gamma_{i,R}^{I}\right].
\]
\end{definition}

Finally, we define the projection for randomized communication. Recall in the deterministic case, there are two steps in the projection.  
In the randomized case, since we put $\rho$ in average, we can remove the first step. Then the projection can be defined as follows.
\begin{definition}[Projection]
Let $I \subseteq [n]$ be the set of unrestricted coordinates. For any $i\in I$, the projection on $i$ is to update the density function from $E^{I}$ to $E^{I\setminus\{i\}}$.
\end{definition}

\begin{remark}
We may use different projections for different communication problems. For example, the BPP lifting theorem \cite{goos2017query} used a very different projection because they studied low-discrepancy gadgets. 
We define the projection in such a way because we are working on AND gadgets. 
Given this flexibility, we believe the density increment arguments may provide new applications beyond the information complexity framework.
\end{remark}

Now we introduce three key lemmas in our proof. 

\begin{lemma}\label{error_bias}
Let $\epsilon\in (0,0.0001)$ be a constant. Let $C$ be a deterministic protocol with error $\epsilon$ under the distribution $\calP$. There is a constant $\delta\in(0,0.02)$ (depending only on $\epsilon$) and a set of coordinates $J\subseteq [n]$ with $|J| = \Omega(n)$ such that $\gamma_i^{[n]} \leq \delta$ holds for each $i\in J$.
\end{lemma}

Since $C$ is a protocol with a small error under $\calP$ and $\calR$ is sampled according to $\calQ$ (\textsf{no} instances), we have that, for a random $R\sim\calR$, it is very likely that $R$ is biased towards to no instances. 
Then Lemma \ref{error_bias} can be proved by an average argument. This is a generalization from the deterministic case that $\gamma_{i}^{[n]}=0$ for all $i\in [n]$. The proof of Lemma \ref{error_bias} is deferred to Section \ref{sec:bias lemma} as part of the proof of Lemma \ref{Bias_lemma}.

\begin{lemma}[Projection lemma]\label{Projection_lemma}
Let $\delta\in(0,0.02)$ be a constant. For any $I\subseteq [n]$ and $i\in I$, if $\gamma_i^{I}\leq \delta$, the projection on $i$ increases the density function by $\Omega(1/k)$, i.e.,
\[
E^{I\setminus\{i\}} \geq E^{I} + \Omega(1/k).
\]
\end{lemma}

The projection lemma shows that the density function increases if we do a projection on a biased coordinate. We prove it in Section \ref{proof_projection_lemma}. 

Our last lemma shows that the bias is preserved during the projections as a counterpart to Lemma \ref{fact: k_u_m} in the deterministic case.

\begin{lemma}[Bias lemma]\label{Bias_lemma}
Let $\delta>0$ be the constant and $J\subseteq [n]$ be the set from Lemma \ref{error_bias}. For any $I\subseteq[n]$ and distinct $i,i'\in I\cap J$, we have that
\[
\gamma_{i'}^{I\setminus\{i\}} \leq \delta.
\]
\end{lemma}

This lemma can be by a convexity inequality and its proof is deferred to the Section \ref{sec:bias lemma}. Now we summarize these three lemmas and complete the proof of Theorem \ref{ic_cost}.
\begin{itemize}
    \item Lemma \ref{error_bias} shows that, if $C$ is a communication protocol with a small error under $\calP$, then $\gamma_{i}^{[n]}$ is very small for many coordinates $i$.
    \item Projection lemma (Lemma \ref{Projection_lemma}) converts the bias on $i$ into the density increment of projection on the coordinate $i$.
    \item Bias lemma (Lemma \ref{Bias_lemma}) proves that a projection on a coordinate $i$ preserves the bias on other coordinates $i'$, which shows the projection lemma can be applied many times.
\end{itemize}

\begin{proof}[Proof of Theorem \ref{ic_cost}]
Assume $C$ has error $\epsilon\in(0,0.0001)$ under distribution $\calP$. By Lemma \ref{error_bias}, there is a constant $\delta\in(0,0.02)$ and a set of coordinates $J\subseteq [n]$ with $|J| = \Omega(n)$ such that $\gamma_i^{[n]} \leq \delta$ for every $i\in J$.

Let $I=[n]\setminus J$. Then iteratively applying Lemma \ref{Projection_lemma} and Lemma \ref{Bias_lemma} on coordinates in $J$, we have 
\[
E^{I} \geq E^{[n]} + \Omega(|J|/k) = E^{[n]} + \Omega(n/k).
\]

By the definition of the density function, we know $E^{[n]} =  -\mutualent(\calQ,\calW:\calR)$.
Since $-E^{I}$ is always non-negative, we have
\[
\mutualent(\calQ,\calW:\calR)=-E^{[n]} \geq -E^{I} + \Omega(n/k) \geq \Omega(n/k).
\]
\end{proof}

\subsection{Proof of the projection lemma}
\label{proof_projection_lemma}

Now we prove Lemma \ref{Projection_lemma}. Recall that 
\[
E^{I}= 
\underset{(\rho,R)\sim (\calW,\calR)}{\mathbb{E}}\left[\ent(\calQ_I,W_I \mid \rho, \calR=R) - \ent(\calQ_I,\calW_I\mid \rho)\right]
\]
and 
\[
\gamma_{i}^{I}= \underset{R\sim\calR}{\mathbb{E}}\left[\gamma_{i,R}^{I}\right].
\]
We aim to show that if $\gamma_i^{I}\leq \delta$ for some $\delta\in(0,0.02)$, then we have that 
\[
E^{I\setminus\{i\}} \geq E^{I} + \Omega(1/k).
\]

In our proof, we borrow a useful lemma from \cite{Gro} and \cite{Jay}. 
In \cite{Gro, Jay}, this lemma was used to analyze information cost.
\begin{lemma}[{\cite[Theorem 3.16]{Gro}}]\label{and_ic}
Let $\delta < 0.02$ be a constant and $I\subseteq[n]$.
Fix a deterministic protocol $C$.
If $\gamma^I_i \leq \delta$, then 
\[
\ent(\calQ_i| \calW_i) -  \ent(\calQ_i| \mathcal{R},\calQ_{I\setminus i},\calW)= \Omega(1/k).
\]
\end{lemma}

Though Lemma \ref{and_ic} is not exactly the same as \cite{Gro,Jay}, the proof is similar and we omit the proof of this lemma here.
We will include the proof in our full version.

\begin{proof}[Proof of Lemma \ref{Projection_lemma}]
Recall 
\[
E^{I} = \ent(\calQ_I,\calW_I \mid \calW_{I^c},\calR) - \ent(\calQ_I,\calW_I\mid \calW_{I^{c}}).
\]

Since $W_{I^c}$ is independent with $(\calQ_I,W_I)$, we have
\[
E^{I} = \ent(\calQ_I,\calW_I \mid \calW_{I^c},\calR) - \ent(\calQ_I,\calW_I).
\]

Similarly $\ent(\calQ_I,\calW_I) - \ent(\calQ_{I\setminus i},\calW_{I\setminus \{i\}}) = \ent(\calQ_i,W_i)$ since $(\calQ_i,\calW_i)$ and $(\calQ_{I\setminus \{i\}},\calW_{I\setminus i})$ are independent. Hence,
\[
E^{I\setminus \{i\}} - E^{I} = \ent(\calQ_i, \calW_i) - \ent(\calQ_I,\calW_I \mid \calW_{I^c},\calR) +   \ent(\calQ_{I\setminus \{i\}},\calW_{I\setminus \{i\}} \mid \calW_{I^c},\calW_i,\calR).
\]

Applying chain rule of entropy on $\ent(\calQ_I,\calW_I|\mathcal{R},\calW_{I^c})$, i.e., 
\[
\ent(\calQ_I,\calW_I|\mathcal{R},\calW_{I^c}) = \ent(\calW_i|\mathcal{R},\calW_{I^c}) + \ent(\calQ_{I\setminus \{i\}},\calW_{I\setminus \{i\}}|\mathcal{R},\calW_i,\calW_{I^c}) + \ent(\calQ_i|\mathcal{R},\calQ_{I\setminus \{i\}},\calW),
\]
we have
\[
E^{I\setminus i} -E^{I} = \ent(\calQ_i,\calW_i) -  \ent(\calW_i|\mathcal{R},\calW_{I^c}) + \ent(\calQ_i|\mathcal{R},\calQ_{I\setminus \{i\}},\calW).
\]

By the chain rule $\ent(\calQ_i,\calW_i) = \ent(\calW_i) + \ent(\calQ_i|\calW_i)$ and the fact that $\ent(\calW_i) \geq \ent(\calW_i|\mathcal{R},\calW_{I^c})$, we conclude that,
\[
E^{I\setminus \{i\}} - E^{I} \geq \ent(\calQ_i|\calW_i) - \ent(\calQ_i|\mathcal{R},\calQ_{I\setminus \{i\}},\calW).
\]

Finally, by Lemma \ref{and_ic} and the fact that $\gamma^I_i \leq \delta<0.02$, we have 
\[
E^{I\setminus \{i\}}- E^{I}\ge\Omega(1/k).
\]
\end{proof}

\section{Deterministic lower bounds for sparse unique-disjointness}
\label{sec: sparse_DCC}

In this section, We discuss the sparse unique-disjointness problem.

\begin{definition}
For each $s\geq 2$ and $n\geq 1$, the $s$-UDISJ problem is defined as follows:
\begin{itemize}
    \item \textsf{No} instances: $D_{0}^{(s)}:=\{(x,y): |x|, |y| \leq s\text{ and }\forall i, x(i) + y(i)\leq 1 \}$.
    \item \textsf{Yes} instances: $D_{*}^{(s)}:=\{(x,y): |x|, |y| \leq s \text{ and }\exists \ell, x(\ell) = y(\ell)=1\text{ and } \forall i\neq \ell, x(i) + y(i)\leq 1 \}$.
\end{itemize}
Here $|x|$ is the Hamming weight of $x$.
\end{definition}

Theorem \ref{thm: sdisk_dcc} aims to show that any deterministic communication protocol for $s$-UDISJ requires $\Omega(s\cdot\log(n/s))$ communication bits. 
To prove this theorem, we consider the following Unique-Equality problem \cite{ST13,LMequality}.
\begin{definition}
Let $s\geq 2$ and $n\geq 1$ be integers. Let $B$ be a set with $(n/s)$ elements. The $s$-UEQUAL problem is defined as follows:
\begin{itemize}
    \item \textsf{No} instances: $B_{0}^{(s)}:=\{(x,y)\in B^s\times B^s:  \forall i \in [s], x_i \neq y_i\}$.
    \item \textsf{Yes} instances: $B_{*}^{(s)}:=\{(x,y)\in B^s\times B^s: \exists \ell, x_{\ell} = y_{\ell} \text{ and } \forall i\in [s]\setminus \{\ell\}, x_i \neq y_i \}$.
\end{itemize}
\end{definition}
There is a simple reduction from $s$-UEQUAL to $s$-UDISJ \cite{ST13}. Hence it is sufficient to prove a communication lower bound for $s$-UEQUAL. In Theorem \ref{thm: sdisk_dcc}, we focus on the regime that $s\leq n^{1/2 - \epsilon}$ for any small constant $\epsilon>0.$ Now our goal is to prove the communication complexity of $s$-UEQUAL is $\Omega(s\cdot\log(n/s)) = \Omega(s\cdot\log n)$.

We borrow the square idea from \cite{LMequality} but revise and simplify it as we do not need to fully simulate the protocol. 
See Section \ref{sec: comparison} for discussions.
\begin{definition}[Square]
Let $R=X\times Y\subseteq B^{s}\times B^{s}$ be a rectangle. A square in $R$ contains a set $I\subseteq [s]$, a set $S\subseteq B^{I}$, and for every $i\in [s]\setminus I$, there is a set $A_i$. 
We denote the family of these $A_i$'s as $\mathcal{A}$.

Given $(I,S,\calA)$, we say it is a square in $R=X\times Y$ if, for every $z\in S$, there exists some $x\in X$ and $y\in Y$ such that: 
\begin{itemize}
    \item $x|_{I}=z$ and, for all $i\in [s]\setminus I$, $x_{i}\in A_{i}$;
    \item $y|_{I}=z$ and, for all $i\in [s]\setminus I$, $y_{i}\in B\setminus A_{i}$.
\end{itemize}
\end{definition}
Same as in previous sections, we use the set $I$ to denote unrestricted coordinates and use $[s]\setminus I$ to denote fixed coordinates. 
We remark that the definition above enforces that $x_{i}\neq y_i$ (as $x_i\in A_i, y_i\in B\setminus A_i$) for all $i\in[s]\setminus I$. Hence, the fixed coordinates do not reveal any information about whether it is a \textsf{yes} instance or a \textsf{no} instance.

Similar to Raz-McKenzie simulation, we also have a notion of thickness in the proof.
\begin{definition}[Thickness]
A set $S\subseteq B^{I}$ is $r$-thick if, it is not empty and for every $i\in I$ and $x\in S$, we have that
\[
|\{x'\in S: \forall j\neq i, x_{j}=x_{j}'\}|\geq r.
\]
We say that a square $(I,S,\calA)$ is $r$-thick if the set $S$ is $r$-thick.
\end{definition}
In our proof, we always choose \textit{$r=10\cdot \log n$}, and sometimes abbreviate $r$-thick as thick. 
The following thickness to full-range lemma is a standard fact in query-to-communication simulations.

\begin{lemma}
\label{lemma: full-range}
Let $S\subseteq B^{I}$ be a thick set. Then for every $z\in\{0,1\}^{I}$, there is a pair $x,y\in S$ such that
\[
\forall i\in I,~ z_i=1 ~~~\text{ iff }~~~ x_i=y_i.
\]
\end{lemma}
The proof of this lemma will be included in the full version. 
As a byproduct of this lemma, we have the following corollary.
\begin{corollary}\label{lemma:monochromatic}
Let $R$ be a rectangle containing a square $(I,S,\calA)$ such that $I\neq\emptyset$ and $S$ is thick, then $R$ is not monochromatic.
\end{corollary}

\begin{definition}[Average degree]
Let $S\subseteq B^{I}$. For each $i\in I$, we define the set $S_{-i}\subseteq B^{I\setminus\{i\}}$ as
\[
S_{-i}:=\{x'\in B^{I\setminus\{i\}}:\exists x\in S, x|_{I\setminus\{i\}}=x'\}.
\]
We say that the average degree of $S$ is $\alpha$ if $|S|\geq \alpha \cdot |S_{-i}|$ holds for all $i$.
We say that a square $(I,S,\calA)$ has an average degree $\alpha$ if the average degree of $S$ is $\alpha$.
\end{definition}

Regards to the average degree, we have a simple but useful fact.
\begin{fact}
\label{fact: trans}
For $\alpha,\beta>0$. Let $S$ be a set that has average degree $\alpha$. Then for any subset $S'\subseteq S$ of size $|S'|\geq \beta\cdot|S|$, $S'$ has average degree $\alpha\cdot\beta$.
\end{fact}
A crucial component in Raz-McKenzie simulation connecting thickness and the average degree is the thickness lemma. In our proof, we borrow a version from \cite{LMequality}.

\begin{lemma}[Thickness lemma \cite{LMequality}]
\label{lemma: sparsethick}
Let $\alpha,\delta>0$ be parameters. Let $\emptyset\neq I\subseteq [s]$ and $S \subseteq B^{I}$. If $S$ has average degree $\alpha$, then there is a $(\delta\cdot\alpha/s)$-thick set $S'\subseteq S$ of size $|S'|\geq (1-\delta)\cdot |S|$.
\end{lemma}

We also fix \textit{$\delta = 1/2$} and \textit{$\alpha = \sqrt{n}$} together with $r=10\cdot\log n$. Recall that $s\leq n^{1/2-\epsilon}$ for some $\epsilon>0$. In this regime of parameters, we have that $\delta\cdot\alpha/s \geq n^{\epsilon}/2 \geq 10\log n = r$. Hence, as long as we maintain a square with average degree $\alpha$. we are able to apply the thickness lemma.

\begin{lemma}[Projection lemma]
\label{lemma: sparseprojection}
Let $R=X\times Y$ be a rectangle and let $Q=(I,S,\calA)$ be a thick square in $R$. If the set $S$ has size more than $(3\alpha)^{|I|}$,  then there is a square $Q'=(I',S',\calA')$ in $R$ such that
\begin{itemize}
    \item $I'\subseteq I$ and $I'\neq\emptyset$,
    \item $S'$ has average degree $2 \alpha$,
    \item $|S'|\geq 0.9\cdot (3\alpha)^{|I'|-|I|}\cdot|S|$.
\end{itemize}
\end{lemma}
\begin{proof}[Proof sketch]
We prove this lemma by a standard structure-vs-pseudorandomness approach. We first describe the process (Algorithm \ref{alg:finding}) to find the set $I'$ and $\tilde{S}$.
\begin{algorithm}
\caption{\textsc{Finding a set $I'$}}
\label{alg:finding}
\begin{algorithmic}[1]
\STATE \textbf{Input:} A set $I\subseteq [s]$ and a set $S\subseteq B^{I}$ of size $|S|\geq (3\alpha)^{|I|}$
\STATE Let $I'\gets I$ and $\tilde{S}\gets S$
\STATE Let $t\gets 0$
\WHILE{$\exists i\in I', |\tilde{S}|\leq 3\alpha\cdot |\tilde{S}_{-i}|$}
\STATE $I'\gets I'\setminus\{i\}$
\STATE $\tilde{S}\gets \tilde{S}_{-i}$
\STATE $t\gets t+1$
\STATE $i_{t}\gets i$ and $S_{t}\gets \tilde{S}$
\ENDWHILE
\RETURN $(I',\tilde{S})$ and $L=(i_1,\dots,i_{t})$
\end{algorithmic}
\end{algorithm}

We note that the average degree of $\tilde{S}$ is at least $(3\alpha)$, otherwise the algorithm would not stop. Following this algorithm, it is also clear that $|\tilde{S}|\geq |S|\cdot(3\alpha)^{|I'|-|I|}$. 
This implies that $I'\neq\emptyset$ because $|S|>(3\alpha)^{|I|}$.

Now, for each $i\in I\setminus I'$, we randomly pick a set $A_{i}\subseteq B$ by independently including each element with probability $1/2$. 
Let $\calA'=\calA\cup \{A_{i}:i\in I\setminus I'\}$ and $S'\subseteq \tilde{S}$ be those strings $z\in\tilde{S}$ such that, there exists an input $x\in X$ and $y\in Y$ such that:
\begin{itemize}
    \item $x|_{I'}=z$ and, for all $i\in [s]\setminus I'$, $x_{i}\in A_{i}$.
    \item $y|_{I'}=z$ and, for all $i\in [s]\setminus I'$, $y_{i}\in B\setminus A_{i}$.
\end{itemize} 

We show that with high probability, the square $(I',S',\calA')$ is a witness for this lemma. We already argued that $|\tilde{S}|\geq |S|\cdot(3\alpha)^{|I'|-|I|},$ now we show for every $z\in\tilde{S}$,
\[
\Pr_{\{A_i\}_{i\in I\setminus I'}}[z\in S'] \geq 1 -O(1/n).
\]
This inequality uses the fact that $S$ is $(10\log n)$-thick, and then a Chernoff bound on each $A_i$, and a union bound on all $i\in I\setminus I'$.
We omit the details here and will include them in the full version. 

Once it is established, then by an average argument, there is a choice of $\{A_{i_{1}},\dots,A_{i_{t}}\}$ such that
\begin{itemize}
    \item $|S'|\geq (1-O(1/n))\cdot |\tilde{S}|\geq 0.9\cdot|S|\cdot (3\alpha)^{|I'|-|I|}$;
    \item $S'$ has average degree $2\alpha$, by Fact \ref{fact: trans} and $\tilde{S}$ having average $3\alpha$ and $|S'|\geq 0.9\cdot |\tilde{S}|$.
\end{itemize}
\end{proof}

Now we are ready to explain how to find a long path in the communication tree.

\subsection{Finding a long path in a communication tree}

Before presenting our algorithm, we first fix some notations.

\begin{definition}
Let $Q=(I,S,\calA)$ be a square in a rectangle $R$. For any sub-rectangle $R'=X'\times Y'$ of $R$, sub-square $Q|_{R'}=(I',S',\calA')$ is defined as follows:
\begin{itemize}
    \item Keep $I'=I$ and $\calA' = \calA$ the same.
    \item $S'\subseteq S$ contains all of those $z\in S$ such that, there exists inputs $x\in X'$ and $x\in Y'$,
    \subitem $x|_{I'}=z$ and, for all $i\in [s]\setminus I'$, $x_{i}\in A_{i}$,
    \subitem $y|_{I'}=z$ and, for all $i\in [s]\setminus I'$, $y_{i}\in B\setminus A_{i}$.
\end{itemize}
\end{definition}

\begin{definition}[Density function]
For a square $Q=(I,S,\calA)$, we define its density as 
\[
E(Q) = \log\left(\frac{|S|}{|B|^{|I|}}\right).
\]
\end{definition}

\begin{algorithm}[t]
	\caption{\textsc{Finding a Long Path}}
	\label{algorithm:simulation}
	\begin{algorithmic}[1]
		\STATE Initialize $v\gets$ root of communication tree $\Pi$ and square $Q_0\gets ([s], B^{s},\emptyset)$
		\STATE Set $t\gets 0$
		\WHILE {$R_v$ is not a monochromatic rectangle}
		\STATE Let $Q_t= (I_t,S_t,\calA_t)$ be currently maintained square.
        \STATE Let $v_0,v_1$ be the children of $v$ in $\Pi$.
		\IF{Alice sends a bit at $v$}
		\STATE Let $X_{v_0},X_{v_1}$ be the partition of $X_{v}$ according to Alice's partition.
		\STATE Let $R_{v_0} \gets X_{v_0} \times Y_v$ and $R_{v_1} \gets X_{v_1} \times Y_v$.
		\ENDIF
		\IF{Bob sends a bit at $v$}
		\STATE Let $Y_{v_0},Y_{v_1}$ be the partition of $Y_{v}$ according to Bob's partition.
		\STATE Let $R_{v_0} \gets X_{v} \times Y_{v_0}$ and $R_{v_1} \gets X_{v} \times Y_{v_1}$.
		\ENDIF
		\IF {$E(Q_{t}|_{R_{v_{0}}}) \geq E(Q_{t}|_{R_{v_{1}}})$} \label{line: simulation}
		    \STATE Update $v \gets  v_0$ and $Q'_{t}\gets Q_{t}|_{R_{v_{0}}}$
		\ELSE
		    \STATE Update $v \gets  v_1$ and $Q'_{t}\gets Q_{t}|_{R_{v_{1}}}$
		\ENDIF
	    \STATE Let $\tilde{Q}_{t}$ be a $r$-thick square obtained by applying Lemma \ref{lemma: sparsethick} on $Q'_{t}$
		\IF {the average degree of $\tilde{Q}_{t}$ is smaller than $(2\alpha)$}
		\STATE Let $Q_{t+1}$ be the square by applying Lemma \ref{lemma: sparseprojection} on $\tilde{Q}_{t}$ 
		\label{line: projection}
		\ELSE
		\STATE Let $Q_{t+1}\gets \tilde{Q}_{t}$
		\ENDIF
		\STATE Update $t\gets t+1$
		\ENDWHILE
	\end{algorithmic}
\end{algorithm}

Now we describe how to find a long path in the communication tree. 
Recall that every node in a communication tree has an associated rectangle. 
Starting from the root, we find a path as follows: 
\begin{enumerate}
\item We maintain a square in each intermediate node.
\item For each intermediate node, the path always visits the left or right child whose associated rectangle maximizes the density. 
\end{enumerate}
The pseudo-code is given in Algorithm \ref{algorithm:simulation}.

\begin{proof}[Proof sketch of Theorem \ref{thm: sdisk_dcc}]
\label{proof: sparse}
Let $t^{*}$ be the value of $t$ when we terminate the algorithm. We note this is the length of our path, which is a lower bound of the deterministic communication complexity. 
Now we argue a lower bound of $t^{*}$ by analyzing the changes to the density function in Algorithm \ref{algorithm:simulation}. We consider two types of density changes, which are called simulation and projection respectively.
\begin{itemize}
    \item \textit{Simulation}. In each round $t$, we obtain a square $\tilde{Q}_{t}$ from the square $Q_t$. For every $t$, we have that  $|\tilde{S}_{t}|\geq |S'_{t}|/2$ by Lemma \ref{lemma: sparsethick} and $|S'_{t}|\geq |S_t|/2$ by the choice on Line \ref{line: simulation}. Hence, 
    \[
    E(\tilde{Q}_{t})\geq E(Q_{t})-2.
    \]
    \item \textit{Projection}. The Line \ref{line: projection} is a projection. For every $t$, if $z = |I_{t}| -|I_{t+1}|>0$, then 
    \[
    E(Q_{t+1}) \geq E(\tilde{Q}_t) + z\cdot (\log (|B|/(3\alpha)) - 2) \geq  E(\tilde{Q}_t) +  z\cdot\Omega(\log n).
    \]
\end{itemize}

Note that, to apply Lemma \ref{lemma: sparsethick}, we need control over the average degrees; to apply Lemma \ref{lemma: sparseprojection}, we need control over the thickness.
Indeed we will inductively show that the following properties $P_1(t), P_2(t), P_3(t)$ are true for all $t\geq 0$,
\begin{itemize}
    \item $P_{1}(t)$: the average degree of $Q_{t} =(I_t,S_t,\calA_t)$ is at least $2\cdot\alpha$.
    \item $P_{2}(t)$: the average degree of $Q'_{t}=(I'_t,S'_t,\calA'_t)$ is at least $\alpha$.
    \item $P_{3}(t)$: $\tilde{Q}_t=(\tilde{I}_t,\tilde{S}_t,\tilde{\calA}_t)$ is $r$-thick.
\end{itemize}
The base case $P_{1}(0)$ is true because $S_{0} = B^{s}$ and $|B|\geq 2\cdot\alpha$. The rest part can be proved by applying the thickness lemma and the projection lemma alternatively. We skip the proof here and will include it in our full version.

Finally, we observe that we must have that $I_{t^{*}} =\emptyset$ when the algorithm terminates at step $t^*$. Otherwise, it is not a monochromatic rectangle by Corollary \ref{lemma:monochromatic}.

We note that the total density decrease before the algorithm termination is at most $2\cdot t^*$. 
On the other hand, the total density increase before the termination is at least $s\cdot \Omega(\log n)$. This implies
\[
2\cdot t^{*}\geq s\cdot \Omega(\log n),
\]
and the result follows.
\end{proof}

\subsection{Discussions and open problems}
\label{sec: comparison}

A very interesting follow-up open problem is to study $s$-UEQUAL lower bounds for $s\gtrsim n^{1/2}$. In our proof, the main bottleneck is Lemma \ref{lemma: sparsethick} (thickness lemma), which requires that $s\le\alpha$.
Note that $\alpha\le|B|$ and $s\cdot|B|=n$. Hence, Lemma \ref{lemma: sparsethick} only applies to the range $s\leq n^{1/2}$. 
In fact, the thickness lemma (or similar lemmas) is also the main barrier in query-to-communication lifting theorems. Lifting theorems usually require a full-range lemma (something similar to Lemma \ref{lemma: full-range}) to maintain a full simulation on the communication tree. 
We use the term full simulation to refer to these proofs aiming to construct a decision tree to exactly compute the Boolean functions.

In contrast, we only attempt to find a long path in the communication tree. 
This approach was suggested by Yang and Zhang \cite{Yang2022SimulationMI}. 
In our analysis, only Corollary \ref{lemma:monochromatic} (a direct corollary of the full-range lemma) is needed. 
This is much weaker than the full-range requirement. Recall that the full-range lemma shows that: For every $z\in\{0,1\}^{I}$, there is a pair $x,y\in S$ such that
\[
\forall i\in I,~ z_i=1 ~~~\text{ iff }~~~ x_i=y_i.
\]
For the $s$-UEQUAL problem, we only care about a subset $\{0^{I},e_1,\dots,e_{|I|}\}$ of $\{0,1\}^{I}$.
Here $e_i\in\{0,1\}^{I}$ is the indicate vector. This observation may give a chance to avoid the full-range barrier, providing tight lower bounds for any $s\geq 1$.

Overall, we believe that finding a long-path paradigm may provide more applications beyond the full simulation paradigm.

\normalem
\bibliographystyle{alpha}
\bibliography{reference}

\appendix
\section{Missing proofs in Section \ref{sec: rcc_vs_udisj}} \label{sec:bias lemma}

In this section, we give a proof for the bias lemma (Lemma \ref{Bias_lemma}). We first recall this lemma below.
\begin{lemma}[Lemma \ref{Bias_lemma} restated]
\label{lemma: 6.1}
Let $\delta>0$ be the constant and $J\subseteq [n]$ be the set from Lemma \ref{error_bias}. For any $I\subseteq[n]$ and distinct $i,i'\in I\cap J$, we have that
\[
\gamma_{i'}^{I\setminus\{i\}} \leq \delta.
\]
\end{lemma}

The proof of Lemma \ref{lemma: 6.1} relies on the following inequality from \cite{Yang2022SimulationMI}.
\begin{lemma}\cite{Yang2022SimulationMI}\label{convexityinequality}
For any $x_1,...,x_n \geq 0$ and $y_1,...,y_n \geq 0$,
\[
\frac{1}{\sum_{j=1}^n y_j}\cdot \sum_{j=1}^n \frac{y_j\cdot x_j}{x_j+y_j} \leq \frac{\sum_{j=1}^n x_j}{\sum_{j=1}^n (x_j+y_j)}.
\]
\end{lemma}

\begin{proof}[Proof of Lemma \ref{lemma: 6.1}]
We first recall the random variable $\ell$ in the definition of $\calP$ (Definition \ref{hard_distribution}). Then for any $i\in [n]$, we have that
\[
\Pr_{x\sim\calP}[x\in D_0 \mid \ell = i] =\Pr_{x\sim\calP}[x\in D_i \mid \ell = i]  = 1/2
\]

For an $i\in [n]$, we call a rectangle $R\in \cal{R}$ good for $i$ if 
\[
\Pr_{x\sim\calP}[x\in D_i \mid x\in R, \ell = i] \leq 0.01.
\]

Since the error of deterministic protocol $C$ under $\calP$ is at most $\epsilon$, we have
\[
\Pr_{x\sim \calP}[C(x) \neq \kdisj(x)]  \leq \epsilon.
\]

Let $\calR_0$ be the set of leaf rectangles that protocol $C$ output $0$ and $\calR_1$ be the set of leaf rectangles that protocol $C$ output $1$, we have
\[
\Pr_{x\sim \calP}[C(x) \neq \kdisj(x)] = \Pr_{x\sim \calP}[x\in \calR_0, x\in D_{*}] + \Pr_{x\sim \calP}[x\in \calR_1, x\in D_{0}] \leq \epsilon.
\]

Since $\Pr_{x\sim\calQ}[x\in R] = \Pr_{x\sim \calP}[x\in R|x\in D_0] \leq 2\cdot \Pr_{x\sim \calP}[x\in R,x\in D_0]$, we have,
\[
\sum_{R\in \calR_1} \Pr_{x\sim \calQ}[x\in R] \cdot \Pr_{x\sim \calP}[x\in D_{*}|x\in R] \leq \sum_{R\in \calR_1} 2\cdot \Pr_{x\sim \calP}[x\in R,x\in D_0] = 2\cdot \Pr_{x\sim \calP}[x\in \calR_1,x\in D_0]
\]

Since $\Pr_{x\in \calQ}[x\in R] = \Pr_{x\in \calP}[x\in R|x\in D_0] \leq 2\cdot \Pr_{x\in \calP}[x\in R]$, we have,
\[
\sum_{R\in \calR_0} \Pr_{x\sim \calQ}[x\in R] \cdot \Pr_{x\sim \calP}[x\in D_{*}|x\in R] \leq 2\cdot  \sum_{R\in \calR_0} \Pr_{x\sim \calP}[x\in R,x\in D_{*}] = 2\cdot \Pr_{x\sim \calP}[x\in \calR_0,x\in D_{*}]
\]

Thus,
\begin{align*}
&\sum_{R\in \calR} \Pr_{x\sim \calQ}[x\in R] \cdot \Pr_{x\sim \calP}[x\in D_{*}|x\in R]\\
=& \sum_{R\in \calR_0} \Pr_{x\sim \calQ}[x\in R] \cdot \Pr_{x\sim \calP}[x\in D_{*}|x\in R] + \sum_{R\in \calR_1} \Pr_{x\sim \calQ}[x\in R] \cdot \Pr_{x\sim \calP}[x\in D_{*}|x\in R]\\
\leq & 2\cdot \Pr_{x\sim \calP}[x\in \calR_0,x\in D_{*}] + 2\cdot \Pr_{x\sim \calP}[x\in \calR_1,x\in D_0] \\
\leq & 2\cdot \epsilon
\end{align*}

Since
\[
\sum_{R\in \calR} \Pr_{x\sim \calQ}[x\in R] \cdot \Pr_{x\sim \calP}[\ell = i|x\in R]\Pr_{x\sim \calP}[x\in D_{i}|x\in R,\ell= i] = \sum_{R\in \calR} \Pr_{x\sim \calQ}[x\in R] \cdot \Pr_{x\sim \calP}[x\in D_{*}|x\in R] \leq 2\cdot \epsilon 
\]

By average argument, there is a set of coordinates $J\subseteq[n]$ with $|J|=\Omega(n)$ such that for any $i\in J$,
\[
\underset{R\sim \calR,i\sim \ell}{\Pr}[R \text{ is good for } i] \geq 1-0.01
\]

For any $(x^{*},\rho^{*},j^{*})$,
\[
p(x^{*},\rho^{*},j^{*}) := \Pr_{(x',\rho',j') \sim (\calQ_{I\setminus \{i',i\}},\calW_{I^c},\calW_i\mid R)}[(x',\rho',j') =(x^{*},\rho^{*},j^{*}) ]
\]
\[
r(x^{*},\rho^{*},j^{*}):= \Pr_{(x',\rho',j')\sim (\calP,\calW_{I^c},\calW_i\mid R)}\left[x'\in D_0, x'\mid_{I\setminus \{i,i'\}} =x^{*},\rho'=\rho^{*},j' = j^{*} \mid x'\notin \bigcup_{l\in I^c}D_{l}\right].
\]
\text{ and }
\[
s(x^{*},\rho^{*},j^{*}):= \Pr_{(x',\rho',j')\sim (\calP,\calW_{I^c},\calW_i\mid R)}\left[x'\in D_{i'}, x'\mid_{I\setminus \{i,i'\}} =x^{*},\rho'=\rho^{*},j' = j^{*} \mid x'\notin \bigcup_{l\in I^c}D_{l}\right].
\]
Intuitively, $p(x^{*},\rho^{*},j^{*})$ is the probability of $(x^{*},\rho^{*},j^{*})$ happens in distribution $(\calQ_{I\setminus \{i',i\}},\calW_{I^c},\calW_i\mid R)$,  $r(x^{*},\rho^{*},j^{*})$ is the probability of $x'\in D_{i'}$ and $(x^{*},\rho^{*},j^{*})$ happens in distribution $(\calP,\calW_{I^c},\calW_i\mid R)$, $s(x^{*},\rho^{*},j^{*})$ is the probability of $x'\in D_{0}$ and $(x^{*},\rho^{*},j^{*})$ happens in distribution $(\calP,\calW_{I^c},\calW_i\mid R)$.

We recall connections between $r_j, s_j$ and $p_j$, let $\rho^{*}_{j^{*}} = (\rho,w_i = j^{*})$,
\[
p(x^{*},\rho^{*},j^{*}) = \frac{r(x^{*},\rho^{*},j^{*})}{\sum r(x^{*},\rho^{*},j^{*})}
\]
\[
\gamma_{i',\rho^{*}_{j^{*}},R}^{I\setminus \{i\}}(x^{*}) = \frac{s(x^{*},\rho^{*},j^{*})}{s(x^{*},\rho^{*},j^{*})+r(x^{*},\rho^{*},j^{*})}
\]
and
\[
\Pr_{x\sim\mu}[x\in D_{i'} \mid x\in R, \ell = i] = \frac{\sum s(x^{*},\rho^{*},j^{*})}{\sum (s(x^{*},\rho^{*},j^{*})+r(x^{*},\rho^{*},j^{*}))}.
\]

By Lemma \ref{convexityinequality}, we have
\[
\gamma_{i',R}^{I\setminus \{i\}} = \sum_{(x^{*},\rho^{*},j^{*})}p(x^{*},\rho^{*},j^{*}) \cdot \gamma_{i',\rho^{*}_{j^{*}},R}^{I\setminus \{i\}}(x^{*}) \leq \Pr_{x\sim\mu}[x\in D_{i'} \mid x\in R, \ell = i] \leq 0.01
\]

Since $i'\in J$, ${\Pr}[\calR \text{ is good for } i'] \geq 1-0.01$. Thus,
\[
\gamma_{i'}^{I\setminus\{i\}} = \sum_{R} \Pr[\calR = R] \cdot \gamma_{i',R}^{I\setminus \{i\}} \leq 1\cdot 0.01 + 0.01\cdot (1-0.01) = \delta
\]

We also can prove that for any $i'\in J$, $\gamma_i^{[n]} \leq \delta$ via replacing $I\setminus \{i\}$ with $[n]$ in above proofs. Thus, $J$ also satisfies Lemma \ref{error_bias}.
\end{proof}

\end{document}